\documentclass[12pt]{article}
\usepackage[latin1]{inputenc}
\usepackage[british]{babel}
\usepackage{cmap}
\usepackage{lmodern}

\usepackage{amssymb, amsmath, amsthm}
\usepackage[a4paper,top=25mm,bottom=25mm,left=25mm,right=25mm]{geometry}
\usepackage{ragged2e}

\usepackage{authblk} 
\usepackage{pifont}
\usepackage{graphicx}
\usepackage[usenames,dvipsnames,svgnames,table]{xcolor}
\usepackage[figuresright]{rotating}
\usepackage{xtab} 
\usepackage{longtable} 
\usepackage{multirow}
\usepackage{footnote}
\usepackage[stable]{footmisc}
\usepackage{chngpage} 
\usepackage{pdflscape} 
\usepackage[nottoc,notlot,notlof]{tocbibind} 

\usepackage{pgfplots}
\pgfplotsset{compat=1.14}
\pgfplotsset{every tick label/.append style={font=\footnotesize}}
\usepgfplotslibrary{fillbetween}
\usepackage{setspace}

\makesavenoteenv{tabular}
\usepackage{tabularx}
\usepackage{booktabs}
\usepackage[flushleft]{threeparttable}
\usepackage[referable]{threeparttablex} 
\newcolumntype{R}{>{\raggedleft\arraybackslash}X}
\newcolumntype{L}{>{\raggedright\arraybackslash}X}
\newcolumntype{C}{>{\centering\arraybackslash}X}
\newcolumntype{M}[1]{>{\centering\arraybackslash}m{#1}}
\newcolumntype{K}{>{\columncolor{gray!20}}C}
\newcolumntype{k}{>{\columncolor{gray!20}}c}

\newlength{\tablen}

\usepackage{dcolumn} 
\newcolumntype{.}{D{.}{.}{-1}}

\usepackage{tikz}
\usetikzlibrary{arrows, calc, matrix, patterns, positioning, trees}
\usepackage[semicolon]{natbib}
\usepackage[hyphens]{url}
\usepackage{hyperref} 
\hypersetup{
  colorlinks   = true,    
  urlcolor     = blue,    
  linkcolor    = blue,    
  citecolor    = ForestGreen      
}
\usepackage{microtype}
\usepackage[justification=centering]{caption} 

\usepackage[labelformat=simple]{subcaption}

\DeclareCaptionLabelFormat{parenthesis}{(#2)}
\makeatletter
\renewcommand\p@subfigure{\arabic{figure}.}
\makeatother

\DeclareCaptionLabelFormat{parenthesis}{(#2)}
\makeatletter
\renewcommand\p@subtable{\arabic{table}.}
\makeatother

\usepackage{enumitem}
\setlist[itemize]{leftmargin=2.5\parindent}
\setlist[enumerate]{leftmargin=2.5\parindent}

\def\addlegendimage{\csname pgfplots@addlegendimage\endcsname}

\theoremstyle{plain}

\newtheorem{corollary}{Corollary}[section]

\newtheorem{proposition}{Proposition}

\theoremstyle{definition}

\newtheorem{example}{Example}

\theoremstyle{remark}

\makeatletter
\let\@fnsymbol\@alph
\makeatother

\def\keywords{\vspace{.5em} 
{\noindent \textit{Keywords}: }}

\def\JEL{\vspace{.5em} 
{\noindent \textbf{\emph{JEL} classification number}: }}

\def\AMS{\vspace{.5em} 
{\noindent \textbf{\emph{MSC} class}: }}

\author{
\href{https://sites.google.com/view/laszlocsato}{L\'aszl\'o Csat\'o}\thanks{~Corresponding author. E-mail: \emph{laszlo.csato@sztaki.hu} \newline Institute for Computer Science and Control (SZTAKI), E\"otv\"os Lor\'and Research Network (ELKH), Laboratory on Engineering and Management Intelligence, Research Group of Operations Research and Decision Systems, Budapest, Hungary \newline 
Corvinus University of Budapest (BCE), Department of Operations Research and Actuarial Sciences, Budapest, Hungary}
$\qquad \qquad$
\href{https://sites.google.com/view/doragretapetroczy}{D\'ora Gr\'eta Petr\'oczy}\thanks{~E-mail: \emph{doragreta.petroczy@uni-corvinus.hu} \newline Corvinus University of Budapest (BCE), Department of Finance, Budapest, Hungary}
}

\title{Fairness in penalty shootouts: \\ Is it worth using dynamic sequences?}
\date{\today}

\def\Dedication{
{\noindent
$\mathfrak{Alle}$ $\mathfrak{diese}$ $\mathfrak{Bestimmungen}$ $\mathfrak{lassen}$ $\mathfrak{sich}$ $\mathfrak{nicht}$ $\mathfrak{absolut}$ $\mathfrak{auf}$ $\mathfrak{jeden}$ $\mathfrak{Fall}$ $\mathfrak{anwenden}$, $\mathfrak{aber}$ $\mathfrak{sie}$ $\mathfrak{m\ddot{u}ssen}$ $\mathfrak{dem}$ $\mathfrak{Handelnden}$ $\mathfrak{gegenw\ddot{a}rtig}$ $\mathfrak{sein}$, $\mathfrak{um}$ $\mathfrak{den}$ $\mathfrak{Nutzen}$ $\mathfrak{der}$ $\mathfrak{in}$ $\mathfrak{ihnen}$ $\mathfrak{enthaltenen}$ $\mathfrak{Wahrheit}$ $\mathfrak{nicht}$ $\mathfrak{zu}$ $\mathfrak{verlieren}$, $\mathfrak{da}$ $\mathfrak{wo}$ $\mathfrak{sie}$ $\mathfrak{gelten}$ $\mathfrak{kann}$.\footnote{~
``\emph{None of them can be applied absolutely in every case, but they must always be present to the mind of the chief, in order that the benefit of the truth contained in them may not be lost in cases where that truth can be of advantage.}'' (Source: Carl von Clausewitz: \emph{On War}, Book 2, Chapter 4 [Methodicism]. Translated by Colonel James John Graham, London, N. Tr\"ubner, 1873. \url{http://clausewitz.com/readings/OnWar1873/TOC.htm})}
}
\vspace{0.25cm}

\flushright
\noindent (Carl von Clausewitz: \emph{Vom Kriege})

\vspace{1cm} 
\justify }

\begin{document}

\maketitle
\thispagestyle{empty}
\Dedication

\begin{abstract}
\noindent
The sequence of moves in a dynamic team tournament may distort the ex-ante winning probabilities and harm efficiency. This paper compares seven soccer penalty shootout rules that determine the kicking order, from a theoretical perspective. Their fairness is evaluated in a reasonable model of First Mover Advantage. We also discuss the probability of reaching the sudden death stage. In the case of stationary scoring probabilities, dynamic mechanisms are not better than static rules. However, it is worth compensating the second-mover by making it the first-mover in the sudden death stage. Our work has the potential to impact decision-makers who can guarantee fairer outcomes in dynamic tournaments by a carefully chosen sequence of actions.



\AMS{60J20, 91A80}

\JEL{C44, C72, Z20}
\end{abstract}

\clearpage
\newgeometry{top=25mm,bottom=25mm,left=25mm,right=25mm}

\section{Introduction} \label{Sec1}

Soccer penalty shootouts---used to decide a tied match in a knockout tournament---offer a natural laboratory to test whether teams with equally skilled players have the same probability of winning. Since the same team kicks the first penalty in all rounds, mental pressure might imply that the second-mover has significantly less than 50\% chance to win.

Previous research has shown mixed evidence and resulted in an intensive debate. While most authors find that there is a First Mover Advantage \citep{ApesteguiaPalacios-Huerta2010, Palacios-Huerta2014, DaSilvaMioranzaMatsushita2018, RudiOlivaresShatty2020}, other papers do not report such a problem \citep{KocherLenzSutter2012, ArrondelDuhautoisLaslier2019}. According to \citet{KassisSchmidtSchreyerSutter2021}, the order of kicking does not influence the winning probabilities, but the right to choose the sequence of moves (being preferred by the coin toss) does matter.
This disagreement is probably due to inadequate sample sizes \citep{VandebroekMcCannVroom2018}, thus the straightforward solution would be to design and implement an appropriate field experiment---which would take years and wide support from the decision-makers.

However, almost all stakeholders recognise that penalty shootouts are potentially \emph{unfair}. According to a survey, more than 90\% of coaches and players want to increase the psychological pressure on the other team by kicking the first penalties \citep{ApesteguiaPalacios-Huerta2010}. With the increasing use of information technology in soccer games such as the video assistant referee (VAR), the fairness of penalty shootouts will probably emerge as a topic of controversy in the future.

An alternative mechanism, the Alternating ($ABBA$) rule, has already been trialled in some matches \citep{Csato2021c}. Even though this sequence does not favour any player in a tennis tiebreak \citep{Cohen-ZadaKrumerShapir2018}, and Monte Carlo simulations show that it substantially mitigates the bias in soccer \citep{DelGiudice2019}, there are other suggestions to improve fairness. \citet{Palacios-Huerta2012} has argued to follow the \href{https://en.wikipedia.org/wiki/Thue\%E2\%80\%93Morse_sequence}{Prouhet--Thue--Morse sequence}, where the first $2^n$ moves are mirrored in the next $2^n$, hence the shooting order becomes $A|B|BA|BAAB|BAABABBA\dots$. Recent proposals include the Catch-up rule \citep{BramsIsmail2018}, its variant called the Adjusted Catch-up rule \citep{Csato2021c}, and the Behind-first rule \citep{AnbarciSunUnver2021}.

We contribute to the topic by evaluating several penalty shootout mechanisms (policy options) in a reasonable mathematical model (state of nature), which  reflects the potential advantage of the team kicking the first penalty. In particular, a standard approach from the literature is adopted \citep{VandebroekMcCannVroom2018, LambersSpieksma2021}, namely, the team lagging behind is assumed to score with a lower probability.

The novel results can be summarised as follows:
\begin{itemize}
\item
Dynamic (history-dependent) sequences are investigated first in this setting.
\item
Two prominent dynamic mechanisms, the Catch-up rule \citep{BramsIsmail2018} and the Behind-first rule \citep{AnbarciSunUnver2021} are found to be equally (un)fair (Proposition~\ref{Prop31}).
\item
The fairness of the penalty shootout designs is compared under a wide set of parameters. The idea of \citet{Csato2021c}, reversing the first- and the second-mover at the beginning of the sudden death stage, improves the fairness of dynamic sequences (Section~\ref{Sec33}).
\item
The probability of reaching the sudden death stage is proved to be independent of the shooting order (Proposition~\ref{Prop32}).
\end{itemize}
To summarise, the known dynamic designs have no advantage over the static rules from any perspective, at least in the mathematical model considered here. Consequently, our findings can filter out poor alternative penalty shootout mechanisms and save the resources for theoretically attractive policy options.

The paper proceeds as follows.
Section~\ref{Sec2} presents the penalty shootout rules to be compared, as well as the mathematical formulation of psychological pressure.
The results are detailed in Section~\ref{Sec3}. Section~\ref{Sec31} computes the probability of winning in the sudden death stage. Section~\ref{Sec32} reveals the equivalence of two dynamic sequences from the perspective of fairness, while Section~\ref{Sec33} contains the numerical analysis for a large set of parameters. Section~\ref{Sec34} deals with some questions beyond fairness.
Finally, Section~\ref{Sec4} offers concluding remarks.

\section{Shootout mechanisms and First Mover Advantage} \label{Sec2}

Denote the team that kicks the first penalty by $A$ and the other team by $B$. The penalty shootout consists of five rounds in its regular phase. In each round, both teams kick one penalty.
The shooting order in a round can be 
(1) independent of the outcomes in the previous rounds (\emph{static} sequence);
(2) influenced by the results of preceding penalties (\emph{dynamic} sequence).
The scores are aggregated after the five rounds, and the team which has scored more goals than the other wins the match.
If the scores are level, the \emph{sudden death} stage starts and continues until one team scores a goal more than the other from the same number of penalties.

Any rule that determines the shooting order, some of them presented in Section~\ref{Sec21}, can be regarded as a policy option to be chosen by the decision-maker.
The probability model in Section~\ref{Sec22} represents a reasonable state of nature by describing how psychological pressure affects performance in a penalty shootout.

\subsection{Penalty shootout designs} \label{Sec21}

We examine three static procedures:
\begin{itemize}
\item
\emph{Standard ($ABAB$)} rule: team $A$ kicks the first and team $B$ kicks the second penalty in each round. This is the official penalty shootout design in soccer.
\item
\emph{Alternating ($ABBA$)} rule: the order of the teams alternates, the second round ($BA$) mirrors the first ($AB$), and this sequence continues even in the possible sudden death phase.
\item
\emph{$ABBA|BAAB$} rule: the order in the first two rounds is $ABBA$, which is mirrored in the next two ($BAAB$), and this sequence is repeated.
\end{itemize}

The ``double alternating'' $ABBA|BAAB$ mechanism is considered because it takes us one step closer to the Prouhet--Thue--Morse sequence than the Alternating ($ABBA$) design. In our opinion, it is unlikely that the organisers want to move further along this line.

There are two basic dynamic designs, both of them having two variants:
\begin{itemize}
\item
\emph{Catch-up} rule \citep{BramsIsmail2018}: the first kicking team alternates but the shooting order does not change if the first kicker missed and the second succeeded in the previous round.
\item
\emph{Adjusted Catch-up} rule \citep{Csato2021c}: the first five rounds are designed according to the Catch-up rule, however, team $B$ kicks the first penalty in the sudden death stage (sixth round) regardless of the outcome in the previous round.
\item
\emph{Behind-first} rule \citep{AnbarciSunUnver2021}: the team having less score kicks the first penalty in the next round, and the order of the previous round is mirrored if the score is tied.
\item
\emph{Adjusted Behind-first} rule: the first five rounds are designed according to the Behind-first rule, however, team $B$ kicks the first penalty in the sudden death stage (sixth round) regardless of the outcome in the previous round.
\end{itemize}
The Adjusted Behind-first mechanism applies the idea underlying the Adjusted Catch-up rule for the Behind-first design. According to our knowledge, it is first defined here.

\begin{table}[t!]
\centering
\caption{An illustration of the penalty shootout mechanisms}
\label{Table1}
\rowcolors{3}{gray!20}{} 
\centerline{
    \begin{tabularx}{1.05\textwidth}{lc CCCCC CCCCC|| CCCC} \toprule \hiderowcolors
\multirow{2}{*}{\emph{Mechanism}} & \multirow{2}{*}{Team} & \multicolumn{10}{c||}{Penalty kicks in the regular phase} &  \multicolumn{4}{c}{Sudden death} \\
     &		& 1     & 2     & 3     & 4     & 5     & 6     & 7     & 8     & 9     & 10    & 11    & 12    & 13    & 14 \\ \bottomrule \showrowcolors
    $ABAB$  & \textcolor{red}{\textbf{A}}     & \textcolor{red}{\ding{55}} &       & \textcolor{red}{\ding{52}} &       & \textcolor{red}{\ding{55}} &       & \textcolor{red}{\ding{52}} &       & \textcolor{red}{\ding{55}} &       & \textcolor{red}{\ding{52}} &       & \textcolor{red}{\ding{52}} &  \\
          & \textcolor{blue}{\textbf{B}}     &       & \textcolor{blue}{\ding{52}} &       & \textcolor{blue}{\ding{52}} &       & \textcolor{blue}{\ding{55}} &       & \textcolor{blue}{\ding{55}} &       & \textcolor{blue}{\ding{55}} &       & \textcolor{blue}{\ding{52}} &       & \textcolor{blue}{\ding{55}} \\ \bottomrule
    $ABBA$  & \textcolor{red}{\textbf{A}}     & \textcolor{red}{\ding{55}} &       &       & \textcolor{red}{\ding{52}} & \textcolor{red}{\ding{55}} &       &       & \textcolor{red}{\ding{52}} & \textcolor{red}{\ding{55}} &       &       & \textcolor{red}{\ding{52}} & \textcolor{red}{\ding{52}} &  \\
          & \textcolor{blue}{\textbf{B}}     &       & \textcolor{blue}{\ding{52}} & \textcolor{blue}{\ding{52}} &       &       & \textcolor{blue}{\ding{55}} & \textcolor{blue}{\ding{55}} &       &       & \textcolor{blue}{\ding{55}} & \textcolor{blue}{\ding{52}} &       &       & \textcolor{blue}{\ding{55}} \\ \bottomrule
    $ABBA|BAAB$ & \textcolor{red}{\textbf{A}}     & \textcolor{red}{\ding{55}} &       &       & \textcolor{red}{\ding{52}} &       & \textcolor{red}{\ding{55}} & \textcolor{red}{\ding{52}} &       & \textcolor{red}{\ding{55}} &       &       & \textcolor{red}{\ding{52}} &       & \textcolor{red}{\ding{52}} \\
          & \textcolor{blue}{\textbf{B}}     &       & \textcolor{blue}{\ding{52}} & \textcolor{blue}{\ding{52}} &       & \textcolor{blue}{\ding{55}} &       &       & \textcolor{blue}{\ding{55}} &       & \textcolor{blue}{\ding{55}}      & \textcolor{blue}{\ding{52}} &       & \textcolor{blue}{\ding{55}} &  \\ \bottomrule
    Catch-up & \textcolor{red}{\textbf{A}}     & \textcolor{red}{\ding{55}} &       & \textcolor{red}{\ding{52}} &       &       & \textcolor{red}{\ding{55}} & \textcolor{red}{\ding{52}} &       &       & \textcolor{red}{\ding{55}} & \textcolor{red}{\ding{52}} &       &       & \textcolor{red}{\ding{52}} \\
          & \textcolor{blue}{\textbf{B}}     &       & \textcolor{blue}{\ding{52}} &       & \textcolor{blue}{\ding{52}} & \textcolor{blue}{\ding{55}} &       &       & \textcolor{blue}{\ding{55}} & \textcolor{blue}{\ding{55}} &       &       & \textcolor{blue}{\ding{52}} & \textcolor{blue}{\ding{55}} &  \\ \bottomrule
    Adj.\ Catch-up & \textcolor{red}{\textbf{A}}     & \textcolor{red}{\ding{55}} &       & \textcolor{red}{\ding{52}} &       &       & \textcolor{red}{\ding{55}} & \textcolor{red}{\ding{52}} &       &       & \textcolor{red}{\ding{55}} &       & \textcolor{red}{\ding{52}} & \textcolor{red}{\ding{52}} &  \\
          & \textcolor{blue}{\textbf{B}}     &       & \textcolor{blue}{\ding{52}} &       & \textcolor{blue}{\ding{52}} & \textcolor{blue}{\ding{55}} &       &       & \textcolor{blue}{\ding{55}} & \textcolor{blue}{\ding{55}} &       & \textcolor{blue}{\ding{52}} &       &       & \textcolor{blue}{\ding{55}} \\ \bottomrule
    Behind-first & \textcolor{red}{\textbf{A}}     & \textcolor{red}{\ding{55}} &       & \textcolor{red}{\ding{52}} &       & \textcolor{red}{\ding{55}} &       & \textcolor{red}{\ding{52}} &       &       & \textcolor{red}{\ding{55}} & \textcolor{red}{\ding{52}} &       &       & \textcolor{red}{\ding{52}} \\
          & \textcolor{blue}{\textbf{B}}     &       & \textcolor{blue}{\ding{52}} &       & \textcolor{blue}{\ding{52}} &       & \textcolor{blue}{\ding{55}} &       & \textcolor{blue}{\ding{55}} & \textcolor{blue}{\ding{55}} &       &       & \textcolor{blue}{\ding{52}} & \textcolor{blue}{\ding{55}} &  \\ \bottomrule
    Adj.\ Behind-first & \textcolor{red}{\textbf{A}}     & \textcolor{red}{\ding{55}} &       & \textcolor{red}{\ding{52}} &       & \textcolor{red}{\ding{55}} &       & \textcolor{red}{\ding{52}} &       &       & \textcolor{red}{\ding{55}} &       & \textcolor{red}{\ding{52}} & \textcolor{red}{\ding{52}} &  \\
          & \textcolor{blue}{\textbf{B}}     &       & \textcolor{blue}{\ding{52}} &       & \textcolor{blue}{\ding{52}} &       & \textcolor{blue}{\ding{55}} &       & \textcolor{blue}{\ding{55}} & \textcolor{blue}{\ding{55}} &       & \textcolor{blue}{\ding{52}} &       &       & \textcolor{blue}{\ding{55}} \\ \bottomrule
    \end{tabularx}
}
\end{table}

Table~\ref{Table1} illustrates the seven penalty shootout designs. Since the scores are 2-2 after five rounds, the shootout goes to sudden death, where both teams succeed in the sixth round. However, in the seventh round only team $A$ scores, hence it wins the match.

Note that the four dynamic mechanisms lead to different shooting orders. Team $B$ kicks the first penalty in the third round under the Catch-up rule because both teams score in the second round where team $A$ is the first-mover. On the other hand, the Behind-first rule favours team $A$ in the third round as it is lagging in the number of goals.
Both designs give the first penalty in the sixth round to team $A$ because team $B$ is the first-mover in the fifth round. The Adjusted Catch-up and Behind-first rules compensate team $B$ by kicking first in the sudden death for being disadvantaged in the first round of the shootout.

\subsection{A model of First Mover Advantage} \label{Sec22}

According to \citet[Figure~2A]{ApesteguiaPalacios-Huerta2010}, the first kicking team scores its penalties with a higher probability in all rounds. A possible reason is that most penalties are successful in soccer, thus a player whose team is lagging behind in the number of goals faces greater mental pressure \citep{ApesteguiaPalacios-Huerta2010, VandebroekMcCannVroom2018, LambersSpieksma2021}. Consequently, each player is assumed to have a probability $p$ of scoring, except for the kicker from a team having fewer scores, who succeeds with probability $q \leq p$. 

\begin{example} \label{Examp21}
Consider a penalty shootout that stands at 2-3 with the first-mover in the fourth round lagging behind. Then the 7th penalty is scored with probability $q$, and the $8$th penalty is scored with probability $p$.
\end{example}

Naturally, our framework is only an idealised version of reality because each player may have a different skill level. There is some evidence for this: stronger teams have a higher chance to win a soccer penalty shootout \citep{Krumer2020b, WunderlichBergeMemmertRein2020}, and small but statistically significant talent gaps can be observed between shooters in ice hockey \citep{LopezSchuckers2017}.

Furthermore, the introduction of two consecutive penalties by the same team may lead to other types of psychological pressure. For example, if the goalkeeper defends a penalty, it may influence the probability of scoring the next penalty by the same team since there is no time to ``calm down''. However, in the absence of adequate experiments with alternative shootout designs, it makes no sense to consider such effects.

\section{Results} \label{Sec3}

The fairness of the soccer penalty shootout is usually interpreted such that no team should have an advantage because of winning or losing the coin toss. Therefore, a mechanism is called \emph{fairer} than another if the probability of winning is closer to $0.5$ for two equally skilled teams.

\subsection{The probability of winning in the sudden death} \label{Sec31}

The first five rounds of a shootout can be represented by a finite sequence of binary numbers as each penalty is either scored or missed. But the sudden death has no definite end, thus it is necessary to calculate the winning probabilities in this phase by hand. With a slight abuse of notation, denote by $W(A)$ the winning probability of the team kicking the first penalty in the sudden death.
For the Standard ($ABAB$) rule, $W^S(A) = p(1-q) + \left[ pq + (1-p)(1-p) \right] W^S(A)$, which leads to
\[
W^S(A) = \frac{p(1-q)}{2p - pq - p^2}.
\]

The Alternating ($ABBA$), (Adjusted) Catch-up, and (Adjusted) Behind-first mechanisms coincide in the sudden death where they imply an alternating order of kicking.
The first penalty in the second round of the sudden death is kicked by team $B$, hence $W^R(A) = p(1-q) + \left[ pq + (1-p)(1-p) \right] \left[ 1 -W^R(A) \right]$, that is,
\[
W^R(A) = \frac{1 - p + p^2}{2 - 2p + pq + p^2}.
\]

The sudden death stage of the $ABBA|BAAB$ rule is the most complicated one.
Here there are two different cases:
(a) the winning probability is $W(AA)$ when team $A$ kicks the first penalty in the first two rounds of this phase, which is followed by two rounds with team $B$ being the first kicker; and
(b) the winning probability is $W(AB)$ when team $A$ kicks the first penalty in the first round of the sudden death, continued with two rounds where team $B$ is the first kicker.
$W(AA)$ or its complement should be used if the regular phase of the shootout consists of an odd number of rounds, and $W(AB)$ should be used if this is an even number.

According to our assumption on First Mover Advantage:
(1) in any round, the winning probabilities of the first and second kickers are $p(1-q)$ and $(1-p)p$, respectively; and
(2) the probability of reaching the next round is $pq + (1-p)^2$. Therefore,
\[
W(AA) = \frac{p(1-q) + \left( 1-2p+pq+p^2 \right) p(1-q) + \left( 1-2p+pq+p^2 \right)^2}{1 + \left( 1-2p+pq+p^2 \right)^2},
\]
and
\[
W(AB) = \frac{p(1-q) + \left( 1-2p+pq+p^2 \right) (1-p)p + \left( 1-2p+pq+p^2 \right)^2}{1 + \left( 1-2p+pq+p^2 \right)^2}.
\]

\subsection{The dynamic sequences are equally (un)fair} \label{Sec32}

Despite the relatively simple mathematical model, deriving analytical statements is non-trivial because five rounds of penalties mean $2^{10} = 1024$ different scenarios and the probability of each contains ten items from the set of $p$, $q$, $(1-p)$, and $(1-q)$. In practice, the shootout is finished if one team has scored more goals than the other could score, but this consideration does not decrease significantly the number of cases to be considered. 

The two basic dynamic designs turn out to be equally (un)fair.

\begin{proposition}  \label{Prop31}
The Catch-up and Behind-first rules lead to the same winning probabilities.
\end{proposition}

\begin{proof}
The likelihood of each possible outcome is shown to be the same under both designs. This probability is the product of the individual probabilities in every round.

Assume that the two rules differ in the probability of the $k$th round.
Hence the shooting order of the teams in the $k$th round is different under the two mechanisms, $CD$ for the Catch-up and $DC$ for the Behind-first. 
If the scores of the teams differ at the beginning of the $k$th round, then the team lagging behind has the probability $q$ of scoring, and the other team has the probability $p$ of scoring. The commutative property of multiplication implies that the probability of the $k$th round is independent of the shooting order, which contradicts the assumption above.
To conclude, the teams should be tied at the beginning of the $k$th round and their shooting order should be different under the two mechanisms.

Consider the case when the shooting order in the $(k-1)$th round was $CD$ by the Catch-up rule. Consequently, team $C$ missed, while team $D$ succeeded in the $(k-1)$th round---otherwise the Catch-up rule would imply the alternated order $DC$ in the $k$th round. Therefore, team $D$ was lagging behind (by one goal) at the beginning of the $(k-1)$th round, thus the shooting order in the $(k-1)$th round was $DC$ by the Behind-first rule. It is a contradiction since no team is lagging behind at the beginning of the $k$th round and the order of the previous round was $DC$ under the Behind-first design, thus it should be $CD$ in the $k$th round according to this mechanism.

Consider the case when the shooting order in the $(k-1)$th round was $DC$ by the Catch-up rule. Then there are three possibilities:
\begin{itemize}
\item
\emph{Team $C$ scored and team $D$ failed in the $(k-1)$th round} \\
This contradicts the assumption that the shooting order in the $k$th round is $CD$ under the Catch-up rule as this procedure implies the unchanged order $DC$ in the $k$th round.
\item
\emph{Team $C$ failed and team $D$ scored in the $(k-1)$th round} \\
Since the teams are tied at the beginning of the $k$th round, team $D$ was lagging in the number of goals at the beginning of the $(k-1)$th round. The shooting order in the $(k-1)$th round was $DC$ under the Behind-first rule, contradicting the shooting order $DC$ in the $k$th round by the Behind-first rule as this design implies an alternated order if the score is tied, which is the case for the $k$th round.
\item
\emph{Both teams failed or both teams scored in the $(k-1)$th round} \\
Since the scores were level at the beginning of the $(k-1)$th round, too, the shooting order was $CD$ under the Behind-first rule. The repetition of the arguments above reveals that the teams should have been tied at the beginning of the penalty shootout, which holds trivially, and the shooting order in the first round is different under the two mechanisms, which is impossible.
\end{itemize}

The proof is completed as the assumption yields a contradiction in each possible case.
\end{proof}

\begin{corollary} \label{Col31}
The Adjusted Catch-up and Adjusted Behind-first rules imply the same winning probabilities.
\end{corollary}

\begin{proof}
The Adjusted Catch-up and Adjusted Behind-first rules coincide with the Catch-up and Behind-first rules, respectively, in their regular stage. Both of them give the first penalty of the sudden death to team $B$, and they follow an alternating order in this phase.
\end{proof}

\subsection{The analysis of fairness} \label{Sec33}

The probability of winning can be accurately determined by a computer code for any values of $p$ and $q$. It is carried out by brute force, calculating the likelihood of all the $2^{2n}$ cases where $n$ is the number of rounds preceding the sudden death stage. Even though the computation can be stopped if one team has scored more goals than the other could score, this trick is not worth implementing because the running time is moderated for any reasonable length of the regular phase.

\citet{VandebroekMcCannVroom2018} derive the winning probabilities for the Standard ($ABAB$) and Alternating ($ABBA$) mechanisms, as well as for the Prouhet--Thue--Morse-sequence. In particular, the winning probability of the first team does not depend on the number of rounds in the regular phase under the Standard ($ABAB$) rule \citet[p.~735]{VandebroekMcCannVroom2018}.
\citet{LambersSpieksma2021} find the least unfair static sequence and empirically compute the winning probabilities for relevant values of $p$ and $q$ if $n=5$ and for various values of $n$ if $p=3/4$ and $q=2/3$.
However, neither study consider dynamic sequences.

\begin{table}[t!]
\centering
\caption{The probability in percentages (\%) that team $A$ wins \\ including the sudden death stage ($p = 3/4$ and $q = 2/3$)}
\label{Table2}
\rowcolors{3}{gray!20}{}
    \begin{tabularx}{\textwidth}{l CCCC >{\bfseries}cCCC} \toprule \hiderowcolors
    Mechanism & \multicolumn{8}{c}{Number of rounds} \\ \midrule \showrowcolors
    & 1     & 2     & 3     & 4     & 5     & 6     & 7     & 8 \\ \bottomrule
    $ABAB$  & 57.14 & 57.14 & 57.14 & 57.14 & 57.14 & 57.14 & 57.14 & 57.14 \\
    $ABBA$  & 52.00 & 51.62 & 51.45 & 51.38 & 51.24 & 51.23 & 51.12 & 51.14 \\
    $ABBA|BAAB$ & 51.04 & 50.44 & 50.74 & 51.00 & 50.66 & 50.41 & 50.65 & 50.80 \\
    Catch-up & 52.00 & 51.50 & 51.49 & 51.30 & 51.27 & 51.18 & 51.15 & 51.10 \\
    Adjusted Catch-up & 52.00 & 50.07 & 51.25 & 50.58 & 50.96 & 50.73 & 50.85 & 50.77 \\
    Behind-first & 52.00 & 51.50 & 51.49 & 51.30 & 51.27 & 51.18 & 51.15 & 51.10 \\
    Adjusted Behind-first & 52.00 & 50.07 & 51.25 & 50.58 & 50.96 & 50.73 & 50.85 & 50.77 \\ \toprule
    \end{tabularx}
\end{table}

First, similarly to recent papers \citep{BramsIsmail2018, Csato2021c, LambersSpieksma2021}, the parameters $p=3/4$ and $q=2/3$ are considered as they are close to the empirical scoring probabilities.
Table~\ref{Table2} presents the probability of winning for team $A$, which kicks the first penalty, for the seven penalty shootout mechanisms.


Except for the Standard ($ABAB$) rule, increasing the number of rounds in the regular phase would improve fairness as there remains more scope to balance the advantages between the two teams. Naturally, organising more rounds requires more time and the number of players is also limited.
Crucially, the dynamic rules are not fairer than the already tried Alternating ($ABBA$) rule. However, the minor amendment proposed by \citet{Csato2021c} in the first round of the sudden death stage consistently makes them fairer.

The calculations reflect the theoretical results:
(1) the fairness of the Standard ($ABAB$) rule is not influenced by the number of rounds \citep{VandebroekMcCannVroom2018};
(2) the (Adjusted) Catch-up and the (Adjusted) Behind-first designs lead to the same winning probabilities (Proposition~\ref{Prop31}).

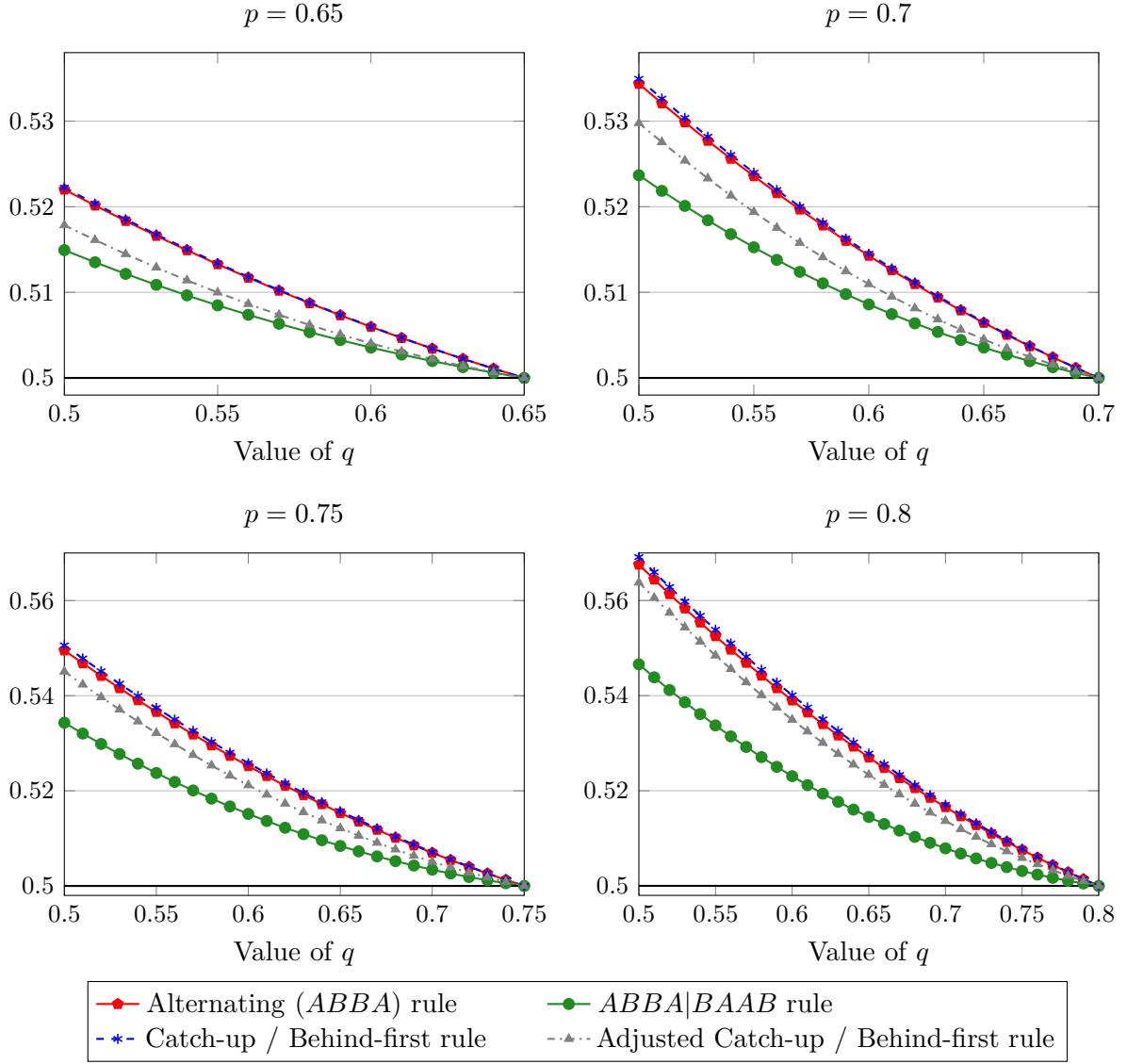
\begin{figure}[t!]
\centering

\begin{tikzpicture}
\begin{axis}[
name = axis1,
title = {$p = 0.65$},
title style = {font=\small},
xlabel = Value of $q$,
x label style = {font=\small},
width = 0.5\textwidth,
height = 0.4\textwidth,
legend style = {font=\small,at={(0.2,-0.15)},anchor=north west,legend columns=6},
ymajorgrids = true,
xmin = 0.5,
xmax = 0.65,
ymin = 0.498,
ymax = 0.538,
max space between ticks=50,
] 

\addlegendentry{Alternating ($ABBA$) rule$\quad$}
\addplot [red, thick, mark=pentagon*] coordinates {
(0.5,0.522002319769121)
(0.51,0.520143825520949)
(0.52,0.518343115894998)
(0.53,0.516599849914361)
(0.54,0.514913659144385)
(0.55,0.513284147786982)
(0.56,0.511710892785259)
(0.57,0.510193443938161)
(0.58,0.50873132402483)
(0.59,0.507324028938377)
(0.6,0.505971027828786)
(0.61,0.504671763254684)
(0.62,0.503425651343704)
(0.63,0.50223208196118)
(0.64,0.501090418886948)
(0.65,0.5)
};

\addlegendentry{$ABBA|BAAB$ rule$\quad$}
\addplot [ForestGreen, thick, mark=oplus*] coordinates {
(0.5,0.514929949580121)
(0.51,0.513512646733746)
(0.52,0.512159687659933)
(0.53,0.510870245179603)
(0.54,0.509643458981578)
(0.55,0.508478433912657)
(0.56,0.507374238311246)
(0.57,0.506329902385564)
(0.58,0.505344416637296)
(0.59,0.504416730331437)
(0.6,0.503545750012982)
(0.61,0.502730338070967)
(0.62,0.501969311350304)
(0.63,0.501261439811687)
(0.64,0.500605445239812)
(0.65,0.5)
};

\addlegendentry{Catch-up / Behind-first rule}
\addplot [blue, thick, dashed, mark=asterisk, mark options={solid,semithick}] coordinates {
(0.5,0.522235289295408)
(0.51,0.520353024973984)
(0.52,0.518528861670871)
(0.53,0.516762648633295)
(0.54,0.515054211494141)
(0.55,0.513403351896949)
(0.56,0.511809847125239)
(0.57,0.510273449736024)
(0.58,0.508793887197391)
(0.59,0.507370861530025)
(0.6,0.506004048952545)
(0.61,0.504693099530565)
(0.62,0.503437636829343)
(0.63,0.502237257569915)
(0.64,0.501091531288639)
(0.65,0.5)
};

\addlegendentry{Adjusted Catch-up / Behind-first rule$\quad$}
\addplot [gray, thick, dashdotted, mark=triangle*, mark options={solid,thin}] coordinates {
(0.5,0.517846867065021)
(0.51,0.516120944918203)
(0.52,0.514471517519388)
(0.53,0.512898743564249)
(0.54,0.511402737876593)
(0.55,0.509983570624287)
(0.56,0.508641266545511)
(0.57,0.50737580418503)
(0.58,0.506187115140164)
(0.59,0.505075083316198)
(0.6,0.504039544190895)
(0.61,0.503080284087895)
(0.62,0.502197039458698)
(0.63,0.501389496172982)
(0.64,0.50065728881703)
(0.65,0.5)
};
\draw [thick] (axis cs:0.5,0.5)  -- (axis cs:0.65,0.5);
\legend{}
\end{axis}

\begin{axis}[
at = {(axis1.south east)},
xshift = 0.1\textwidth,
title = {$p = 0.7$},
title style = {font=\small},
xlabel = Value of $q$,
x label style = {font=\small},
width = 0.5\textwidth,
height = 0.4\textwidth,
legend style = {font=\small,at={(0.2,-0.15)},anchor=north west,legend columns=6},
ymajorgrids = true,
xmin = 0.5,
xmax = 0.7,
ymin = 0.498,
ymax = 0.538,
]      

\addlegendentry{Alternating ($ABBA$) rule$\quad$}
\addplot [red, thick, mark=pentagon*] coordinates {
(0.5,0.534361090944444)
(0.51,0.532080972365286)
(0.52,0.529861873356239)
(0.53,0.527703582617567)
(0.54,0.525605866071145)
(0.55,0.52356846676536)
(0.56,0.521591104787707)
(0.57,0.51967347718482)
(0.58,0.517815257889714)
(0.59,0.51601609765597)
(0.6,0.514275623998675)
(0.61,0.512593441141872)
(0.62,0.510969129972324)
(0.63,0.509402247999379)
(0.64,0.507892329320749)
(0.65,0.506438884594032)
(0.66,0.505041401013749)
(0.67,0.503699342293776)
(0.68,0.502412148654962)
(0.69,0.501179236817808)
(0.7,0.5)
};

\addlegendentry{$ABBA|BAAB$ rule$\quad$}
\addplot [ForestGreen, thick, mark=oplus*] coordinates {
(0.5,0.523680312667559)
(0.51,0.521852045599161)
(0.52,0.520096671542051)
(0.53,0.518413503252399)
(0.54,0.516801822864275)
(0.55,0.515260879697548)
(0.56,0.513789888091693)
(0.57,0.512388025267857)
(0.58,0.511054429221328)
(0.59,0.509788196646413)
(0.6,0.508588380895579)
(0.61,0.50745398997456)
(0.62,0.506383984574949)
(0.63,0.50537727614569)
(0.64,0.50443272500471)
(0.65,0.503549138491819)
(0.66,0.502725269163814)
(0.67,0.501959813032667)
(0.68,0.501251407847488)
(0.69,0.500598631420859)
(0.7,0.5)
};

\addlegendentry{Catch-up / Behind-first rule}
\addplot [blue, thick, dashed, mark=asterisk, mark options={solid,semithick}] coordinates {
(0.5,0.534943233402778)
(0.51,0.532629963800375)
(0.52,0.530376858203575)
(0.53,0.528183897143369)
(0.54,0.526051044911355)
(0.55,0.523978249029385)
(0.56,0.521965439721824)
(0.57,0.520012529390315)
(0.58,0.518119412090994)
(0.59,0.516285963014049)
(0.6,0.514512037965562)
(0.61,0.512797472851561)
(0.62,0.511142083164205)
(0.63,0.509545663470027)
(0.64,0.508007986900195)
(0.65,0.506528804642697)
(0.66,0.505107845436406)
(0.67,0.503744815066962)
(0.68,0.502439395864412)
(0.69,0.501191246202569)
(0.7,0.5)
};

\addlegendentry{Adjusted Catch-up / Behind-first rule$\quad$}
\addplot [gray, thick, dashdotted, mark=triangle*, mark options={solid,thin}] coordinates {
(0.5,0.529773367944444)
(0.51,0.52754171108637)
(0.52,0.525387054597647)
(0.53,0.523309774714019)
(0.54,0.521310210639721)
(0.55,0.519388663252859)
(0.56,0.517545393818491)
(0.57,0.515780622709128)
(0.58,0.514094528132466)
(0.59,0.512487244866046)
(0.6,0.510958862998674)
(0.61,0.509509426678356)
(0.62,0.508138932866533)
(0.63,0.506847330098431)
(0.64,0.505634517249325)
(0.65,0.504500342306532)
(0.66,0.503444601146933)
(0.67,0.502467036319884)
(0.68,0.501567335835314)
(0.69,0.500745131956884)
(0.7,0.5)
};
\draw [thick] (axis cs:0.5,0.5)  -- (axis cs:0.7,0.5);
\legend{}
\end{axis}
\end{tikzpicture}

\vspace{0.25cm}
\begin{tikzpicture}
\begin{axis}[
name = axis3,
title = {$p = 0.75$},
title style = {font=\small},
xlabel = Value of $q$,
x label style = {font=\small},
width = 0.5\textwidth,
height = 0.4\textwidth,
legend style = {font=\small,at={(0.05,-0.25)},anchor=north west,legend columns=2},
ymajorgrids = true,
xmin = 0.5,
xmax = 0.75,
ymin = 0.498,
ymax = 0.57,
]      

\addlegendentry{Alternating ($ABBA$) rule$\quad \qquad$}
\addplot [red, thick, mark=pentagon*] coordinates {
(0.5,0.549530278081478)
(0.51,0.546817682806298)
(0.52,0.544167775084419)
(0.53,0.541580417782331)
(0.54,0.539055458649892)
(0.55,0.536592730030124)
(0.56,0.534192048571415)
(0.57,0.531853214942045)
(0.58,0.529576013546955)
(0.59,0.527360212246667)
(0.6,0.525205562078301)
(0.61,0.523111796978608)
(0.62,0.521078633508929)
(0.63,0.519105770582067)
(0.64,0.517192889190933)
(0.65,0.515339652138986)
(0.66,0.513545703772334)
(0.67,0.511810669713493)
(0.68,0.510134156596719)
(0.69,0.508515751804874)
(0.7,0.506955023207776)
(0.71,0.505451518901965)
(0.72,0.504004766951872)
(0.73,0.502614275132312)
(0.74,0.50127953067228)
(0.75,0.5)
};

\addlegendentry{$ABBA|BAAB$ rule$\qquad \qquad \qquad \qquad$}
\addplot [ForestGreen, thick, mark=oplus*] coordinates {
(0.5,0.534316291183721)
(0.51,0.532041304736522)
(0.52,0.529847430443475)
(0.53,0.527734102405373)
(0.54,0.525700727578527)
(0.55,0.523746683304167)
(0.56,0.521871314833649)
(0.57,0.520073932853401)
(0.58,0.518353811013259)
(0.59,0.516710183461742)
(0.6,0.51514224239161)
(0.61,0.513649135598854)
(0.62,0.512229964058118)
(0.63,0.510883779517351)
(0.64,0.509609582114296)
(0.65,0.508406318017296)
(0.66,0.507272877092641)
(0.67,0.506208090600586)
(0.68,0.505210728921939)
(0.69,0.504279499316993)
(0.7,0.503413043718385)
(0.71,0.502609936559311)
(0.72,0.501868682638399)
(0.73,0.501187715022348)
(0.74,0.500565392987327)
(0.75,0.5)
};

\addlegendentry{Catch-up / Behind-first rule$\qquad$}
\addplot [blue, thick, dashed, mark=asterisk, mark options={solid,semithick}] coordinates {
(0.5,0.550586368726647)
(0.51,0.547836127446923)
(0.52,0.545146540729948)
(0.53,0.542517646785024)
(0.54,0.539949477332748)
(0.55,0.537442056988214)
(0.56,0.534995402645831)
(0.57,0.532609522865759)
(0.58,0.530284417261869)
(0.59,0.528020075891186)
(0.6,0.525816478644756)
(0.61,0.523673594639911)
(0.62,0.521591381613837)
(0.63,0.519569785318464)
(0.64,0.517608738916572)
(0.65,0.51570816237911)
(0.66,0.513867961883683)
(0.67,0.512088029214142)
(0.68,0.510368241161277)
(0.69,0.50870845892455)
(0.7,0.50710852751484)
(0.71,0.505568275158174)
(0.72,0.504087512700404)
(0.73,0.502666033012799)
(0.74,0.501303610398527)
(0.75,0.5)
};

\addlegendentry{Adjusted Catch-up / Behind-first rule}
\addplot [gray, thick, dashdotted, mark=triangle*, mark options={solid,thin}] coordinates {
(0.5,0.545084248418393)
(0.51,0.542348552925829)
(0.52,0.539688139915278)
(0.53,0.537103510339167)
(0.54,0.534595138642665)
(0.55,0.532163470996921)
(0.56,0.529808923534696)
(0.57,0.527531880588335)
(0.58,0.525332692929963)
(0.59,0.523211676013855)
(0.6,0.52116910822088)
(0.61,0.519205229104975)
(0.62,0.517320237641546)
(0.63,0.515514290477769)
(0.64,0.513787500184683)
(0.65,0.512139933511056)
(0.66,0.510571609638936)
(0.67,0.509082498440837)
(0.68,0.507672518738516)
(0.69,0.506341536563273)
(0.7,0.505089363417737)
(0.71,0.503915754539076)
(0.72,0.502820407163591)
(0.73,0.501802958792664)
(0.74,0.500862985459975)
(0.75,0.5)
};
\draw [thick] (axis cs:0.5,0.5)  -- (axis cs:0.75,0.5);
\end{axis}

\begin{axis}[
at = {(axis3.south east)},
xshift = 0.1\textwidth,
title = {$p = 0.8$},
title style = {font=\small},
xlabel = Value of $q$,
x label style = {font=\small},
width = 0.5\textwidth,
height = 0.4\textwidth,
legend style = {font=\small,at={(-0.05,-0.2)},anchor=north west,legend columns=3},
ymajorgrids = true,
xmin = 0.5,
xmax = 0.8,
ymin = 0.498,
ymax = 0.57,
]      

\addlegendentry{Alternating ($ABBA$) rule$\quad$}
\addplot [red, thick, mark=pentagon*] coordinates {
(0.5,0.567566736000001)
(0.51,0.564428166543774)
(0.52,0.561352524172996)
(0.53,0.558339668219226)
(0.54,0.55538945288448)
(0.55,0.552501726756756)
(0.56,0.549676332321362)
(0.57,0.546913105468203)
(0.58,0.544211874995199)
(0.59,0.541572462107946)
(0.6,0.538994679915789)
(0.61,0.536478332924418)
(0.62,0.534023216525119)
(0.63,0.531629116480811)
(0.64,0.529295808408966)
(0.65,0.527023057261538)
(0.66,0.524810616802011)
(0.67,0.522658229079653)
(0.68,0.52056562390109)
(0.69,0.518532518299281)
(0.7,0.516558616)
(0.71,0.514643606885892)
(0.72,0.51278716645822)
(0.73,0.51098895529633)
(0.74,0.509248618514973)
(0.75,0.507565785219511)
(0.76,0.505940067959102)
(0.77,0.50437106217792)
(0.78,0.502858345664491)
(0.79,0.501401477999197)
(0.8,0.5)
};

\addlegendentry{$ABBA|BAAB$ rule$\quad$}
\addplot [ForestGreen, thick, mark=oplus*] coordinates {
(0.5,0.546594327077749)
(0.51,0.543845343084421)
(0.52,0.541185772770148)
(0.53,0.538615127973213)
(0.54,0.53613289578081)
(0.55,0.533738536117034)
(0.56,0.531431479287893)
(0.57,0.529211123488606)
(0.58,0.527076832278307)
(0.59,0.525027932027066)
(0.6,0.523063709340049)
(0.61,0.521183408463394)
(0.62,0.519386228676233)
(0.63,0.517671321673142)
(0.64,0.516037788941007)
(0.65,0.514484679134227)
(0.66,0.513010985451867)
(0.67,0.511615643020283)
(0.68,0.51029752628442)
(0.69,0.509055446410921)
(0.7,0.507888148705882)
(0.71,0.506794310049956)
(0.72,0.505772536353302)
(0.73,0.504821360032661)
(0.74,0.503939237512723)
(0.75,0.503124546753688)
(0.76,0.502375584806824)
(0.77,0.501690565399597)
(0.78,0.501067616551834)
(0.79,0.500504778224182)
(0.8,0.5)
};

\addlegendentry{Catch-up / Behind-first rule}
\addplot [blue, thick, dashed, mark=asterisk, mark options={solid,semithick}] coordinates {
(0.5,0.569138373333334)
(0.51,0.565964049441512)
(0.52,0.562849650052804)
(0.53,0.559795173712438)
(0.54,0.556800623346399)
(0.55,0.553866005621621)
(0.56,0.550991330308492)
(0.57,0.54817660964556)
(0.58,0.545421857706398)
(0.59,0.542727089768505)
(0.6,0.540092321684209)
(0.61,0.537517569253474)
(0.62,0.535002847598527)
(0.63,0.532548170540296)
(0.64,0.530153549976522)
(0.65,0.527818995261538)
(0.66,0.52554451258763)
(0.67,0.523330104367932)
(0.68,0.521175768620807)
(0.69,0.519081498355635)
(0.7,0.51704728096)
(0.71,0.515073097588185)
(0.72,0.513158922550971)
(0.73,0.511304722706653)
(0.74,0.509510456853262)
(0.75,0.50777607512195)
(0.76,0.506101518371472)
(0.77,0.504486717583757)
(0.78,0.502931593260523)
(0.79,0.501436054820897)
(0.8,0.5)
};

\addlegendentry{Adjusted Catch-up / Behind-first rule$\quad$}
\addplot [gray, thick, dashdotted, mark=triangle*, mark options={solid,thin}] coordinates {
(0.5,0.563788944000001)
(0.51,0.56057604369603)
(0.52,0.557434991283908)
(0.53,0.554366284656474)
(0.54,0.551370408221184)
(0.55,0.548447830756756)
(0.56,0.545599003265619)
(0.57,0.542824356822315)
(0.58,0.540124300418048)
(0.59,0.537499218801451)
(0.6,0.534949470315788)
(0.61,0.532475384732674)
(0.62,0.530077261082431)
(0.63,0.527755365481259)
(0.64,0.52550992895527)
(0.65,0.523341145261538)
(0.66,0.521249168706267)
(0.67,0.519234111960165)
(0.68,0.517296043871138)
(0.69,0.515434987274385)
(0.7,0.5136509168)
(0.71,0.511943756678149)
(0.72,0.510313378541932)
(0.73,0.508759599227979)
(0.74,0.507282178574877)
(0.75,0.505880817219511)
(0.76,0.504555154391358)
(0.77,0.503304765704832)
(0.78,0.502129160949739)
(0.79,0.501027781879901)
(0.8,0.5)
};
\draw [thick] (axis cs:0.5,0.5)  -- (axis cs:0.8,0.5);
\legend{}
\end{axis}
\end{tikzpicture}

\caption{The probability that team $A$ wins a penalty shootout \\
over five rounds including the sudden death stage}
\label{Fig1}

\end{figure}


In order to extend these findings, Figure~\ref{Fig1} plots the winning probability of team $A$ as the function of parameter $q$ for four values of $p$. The Standard ($ABAB$) rule is not depicted due to its high level of unfairness.
The main messages can be summarised as follows:
\begin{itemize}
\item
Fairness is difficult to achieve because the disadvantage of team $B$ cannot be balanced by providing it with a higher scoring probability as the first kicker in a later round.
\item
The straightforward Alternating ($ABBA$) rule is not worse than the dynamic sequences, the use of the latter cannot be justified by the need of improving fairness. This finding considerably decreases the value of the contribution by \citet{BramsIsmail2018}.
\item
The adjustment of the dynamic designs, suggested by \citet{Csato2021c}, robustly improves fairness. It is worth guaranteeing the first penalty of the sudden death stage for team $B$. Nonetheless, the effect of this modification decreases when the advantaged kicker becomes more successful ($p$ grows).
\item
The $ABBA|BAAB$ rule is the closest to fairness among all mechanisms considered here. If the Alternating ($ABBA$) rule is judged inadequate, a further step towards the Prouhet--True--Morse sequence can be effective with respect to fairness.
\end{itemize}

\subsection{Further aspects of comparison} \label{Sec34}

In the following, we discuss three other topics: the expected length of the sudden death stage and the probability of reaching it, as well as the strategy-proofness of the mechanisms.

Denote the expected length of the sudden death by $\varepsilon$ and the probability that it finishes in a given round by $R$. Then
\[
\varepsilon = R + (1-R)(1 + \varepsilon)
\]
because the expected length of the sudden death is $1 + \varepsilon$ if it is not decided in the first round. Consequently, $\varepsilon = 1/R$ where $R = p(1-q) + (1-p)p = 2p-pq-p^2$, which does not depend on the mechanism used to determine the shooting order.

The probability of reaching the sudden death stage can be influenced by the penalty shootout design, see \citet[Chapter~4.5]{Csato2021a}. However, this is not the case in our model.

\begin{proposition}  \label{Prop32}
The probability of reaching the sudden death stage is independent of the shooting order.
\end{proposition}

\begin{proof}
We focus on the standing of the penalty shootout after some rounds. The possible states can be distinguished by the difference $k$ between the number of goals scored by the two teams.
The transition probabilities from any round to the next round can be easily computed:
\begin{itemize}
\item
\emph{The score difference is $k=0$} \\
The scores will remain tied if both teams fail in the next round with probability $(1-p)(1-p)$, or both teams succeed in the next round with probability $pq$. Otherwise, the shootout goes to state $k=1$.
\item
\emph{The score difference is $k \neq 0$} \\
The score difference will remain $k$ if both teams fail in the next round with probability $(1-p)(1-q)$, or both teams succeed in the next round with probability $pq$. The shootout goes to state $k+1$ if the team lagging behind misses its penalty, while the other team scores. This has the probability $p(1-q)$, independently of the shooting order. Otherwise, if the team lagging behind scores and the other team misses, the penalty shootout goes to state $k-1$ with probability $(1-p)q$.
\end{itemize}

\begin{table}[t!]
  \centering
  \caption{Transition probabilities between the states of a penalty shootout}
  \label{Table3}
    \begin{tabularx}{\textwidth}{L CCC c} \toprule
    To $\rightarrow$ & \multicolumn{4}{c}{Possible states} \\
    From $\downarrow$ & Score diff.\ $k=0$ & Score diff.\ $k=1$ & Score diff.\ $k=2$ & $\cdots$ \\ \midrule
    Score diff.\ $k=0$ & $pq+(1-p)(1-p)$ & $p(1-q)+(1-p)p$ & 0     & $\cdots$ \\
    Score diff.\ $k=1$ & $(1-p)q$ & $pq+(1-p)(1-q)$ & $p(1-q)$ & $\cdots$ \\
    Score diff.\ $k=2$ & 0     & $(1-p)q$ & $pq+(1-p)(1-q)$ & $\cdots$ \\
    $\vdots$ & $\vdots$ & $\vdots$ & $\vdots$ & $\ddots$ \\ \bottomrule
    \end{tabularx}
\end{table}

Table~\ref{Table3} overviews the possible moves between these states. Since the transition probabilities are independent of the shooting order, the probability that the penalty shootout finishes in the state $k=0$ at the end of its regular phase---and continues with the sudden death stage---is determined only by the parameters $p$ and $q$.
\end{proof}

For instance, it has a probability of $0.215$ that a penalty shootout over five rounds reaches the sudden death stage.

In static mechanisms, the shooting order cannot be controlled by the teams. However, under a dynamic rule, the team kicking the second penalty might gain from deliberately missing it if the rewards outweigh the loss of an uncertain goal, see \citet{BramsIsmail2018}.
In our mathematical model, the scoring probabilities are independent of the shooting position, thus a deliberate miss cannot yield any profit.

\section{Conclusions} \label{Sec4}

Seven soccer penalty shootout rules have been studied in a reasonable model of First Mover Advantage.
Under the stationary scoring probabilities considered here, it remains sufficient to use static rules in order to improve fairness since the Catch-up and the Behind-first designs have no advantage over the $ABBA|BAAB$ mechanism from any perspective. On the other hand, compensating the second-mover by making it first-mover in the sudden death stage seems to be a reasonable modification for the dynamic sequences.

There are obvious extensions to our paper. One can introduce the Adjusted $ABBA$ or Adjusted $ABBA|BAAB$ mechanisms by providing the first penalty for team $B$ in the sudden death phase and attempt to prove their dominance over the unadjusted variant. Since the dynamic mechanisms tend to be more adaptive, they might perform better in non-stationary settings where the scoring probabilities vary across rounds.

The example of the National Football League reinforces that the decision-makers may be keen to reduce a considerable difference in winning probabilities when it is implied by a pure coin toss: the new rules for overtime, introduced in 2010 for playoff games and extended in 2012 to all games, better balance out the coin toss advantage without significantly increasing the expected length of the game \citep{Jones2012, MartinTimmonsPowell2018}.
While the experiments with fairer penalty shootout systems have recently been stopped in soccer, the debate will probably continue as advancing to the next round in a knockout tournament has huge financial implications and sporting effects. In addition, the abolition of the away goals rule in the European club soccer competitions from the 2021/22 season \citep{UEFA2021f} has obviously enhanced the role of penalty shootouts.
The theoretical results above can help to decide what policy options are worthwhile to choose for implementing in the field. Hopefully, our paper might also inspire further research on penalty shootouts.

\section*{Acknowledgements}
\addcontentsline{toc}{section}{Acknowledgements}
\noindent
\emph{L\'aszl\'o Csat\'o}, the father of the first author has contributed to the paper by writing the key part of the code making the necessary computations in Python. \\
We are indebted to \emph{Steven J.\ Brams} and \emph{Mehmet S.\ Ismail}, whose work was a great source of inspiration. \\
We are grateful to \emph{Tam\'as Halm} for useful advice. \\
Seven anonymous reviewers provided valuable comments and suggestions on earlier drafts. \\
The research was supported by the MTA Premium Postdoctoral Research Program grant PPD2019-9/2019, the NKFIH grant K 128573, and the Higher Education Institutional Excellence Program 2020 of the Ministry of Innovation and Technology in the framework of the `Financial and Public Services' research project at Corvinus University of Budapest.

\bibliographystyle{apalike}
\bibliography{All_references}

\begin{thebibliography}{}

\bibitem[Anbarc{\i} et~al., 2021]{AnbarciSunUnver2021}
Anbarc{\i}, N., Sun, C.-J., and {\"U}nver, M.~U. (2021).
\newblock Designing practical and fair sequential team contests: The case of
  penalty shootouts.
\newblock {\em Games and Economic Behavior}, 130:25--43.

\bibitem[Apesteguia and Palacios-Huerta, 2010]{ApesteguiaPalacios-Huerta2010}
Apesteguia, J. and Palacios-Huerta, I. (2010).
\newblock Psychological pressure in competitive environments: Evidence from a
  randomized natural experiment.
\newblock {\em American Economic Review}, 100(5):2548--2564.

\bibitem[Arrondel et~al., 2019]{ArrondelDuhautoisLaslier2019}
Arrondel, L., Duhautois, R., and Laslier, J.-F. (2019).
\newblock Decision under psychological pressure: The shooter's anxiety at the
  penalty kick.
\newblock {\em Journal of Economic Psychology}, 70:22--35.

\bibitem[Brams and Ismail, 2018]{BramsIsmail2018}
Brams, S.~J. and Ismail, M.~S. (2018).
\newblock Making the rules of sports fairer.
\newblock {\em SIAM Review}, 60(1):181--202.

\bibitem[Cohen-Zada et~al., 2018]{Cohen-ZadaKrumerShapir2018}
Cohen-Zada, D., Krumer, A., and Shapir, O.~M. (2018).
\newblock Testing the effect of serve order in tennis tiebreak.
\newblock {\em Journal of Economic Behavior \& Organization}, 146:106--115.

\bibitem[Csat\'o, 2021a]{Csato2021c}
Csat\'o, L. (2021a).
\newblock A comparison of penalty shootout designs in soccer.
\newblock {\em 4OR}, 19(5):183--198.

\bibitem[Csat\'o, 2021b]{Csato2021a}
Csat\'o, L. (2021b).
\newblock {\em Tournament Design: How Operations Research Can Improve Sports
  Rules}.
\newblock Palgrave Pivots in Sports Economics. Palgrave Macmillan, Cham,
  Switzerland.

\bibitem[Da~Silva et~al., 2018]{DaSilvaMioranzaMatsushita2018}
Da~Silva, S., Mioranza, D., and Matsushita, R. (2018).
\newblock {FIFA} is right: The penalty shootout should adopt the tennis
  tiebreak format.
\newblock {\em Open Access Library Journal}, 5(3):1--23.

\bibitem[Del~Giudice, 2019]{DelGiudice2019}
Del~Giudice, P.~E.~S. (2019).
\newblock Modeling football penalty shootouts: How improving individual
  performance affects team performance and the fairness of the {ABAB} sequence.
\newblock {\em International Journal of Sport and Health Sciences},
  13(5):240--245.

\bibitem[Jones, 2012]{Jones2012}
Jones, C. (2012).
\newblock The new rules for {NFL} overtime.
\newblock {\em Mathematics Magazine}, 85(4):277--283.

\bibitem[Kassis et~al., 2021]{KassisSchmidtSchreyerSutter2021}
Kassis, M., Schmidt, S.~L., Schreyer, D., and Sutter, M. (2021).
\newblock Psychological pressure and the right to determine the moves in
  dynamic tournaments -- evidence from a natural field experiment.
\newblock {\em Games and Economic Behavior}, 126:278--287.

\bibitem[Kocher et~al., 2012]{KocherLenzSutter2012}
Kocher, M.~G., Lenz, M.~V., and Sutter, M. (2012).
\newblock Psychological pressure in competitive environments: New evidence from
  randomized natural experiments.
\newblock {\em Management Science}, 58(8):1585--1591.

\bibitem[Krumer, 2020]{Krumer2020b}
Krumer, A. (2020).
\newblock Pressure versus ability: Evidence from penalty shoot-outs between
  teams from different divisions.
\newblock {\em Journal of Behavioral and Experimental Economics}, 89:101578.

\bibitem[Lambers and Spieksma, 2021]{LambersSpieksma2021}
Lambers, R. and Spieksma, F.~C.~R. (2021).
\newblock A mathematical analysis of fairness in shootouts.
\newblock {\em IMA Journal of Management Mathematics}, 32(4):411--424.

\bibitem[Lopez and Schuckers, 2017]{LopezSchuckers2017}
Lopez, M.~J. and Schuckers, M. (2017).
\newblock Predicting coin flips: using resampling and hierarchical models to
  help untangle the {NHL}'s shoot-out.
\newblock {\em Journal of Sports Sciences}, 35(9):888--897.

\bibitem[Martin et~al., 2018]{MartinTimmonsPowell2018}
Martin, R., Timmons, L., and Powell, M. (2018).
\newblock A {M}arkov chain analysis of {NFL} overtime rules.
\newblock {\em Journal of Sports Analytics}, 4(2):95--105.

\bibitem[Palacios-Huerta, 2012]{Palacios-Huerta2012}
Palacios-Huerta, I. (2012).
\newblock Tournaments, fairness and the {P}rouhet-{T}hue-{M}orse sequence.
\newblock {\em Economic Inquiry}, 50(3):848--849.

\bibitem[Palacios-Huerta, 2014]{Palacios-Huerta2014}
Palacios-Huerta, I. (2014).
\newblock {\em Beautiful Game Theory: How Soccer Can Help Economics}.
\newblock Princeton University Press, Princeton.

\bibitem[Rudi et~al., 2020]{RudiOlivaresShatty2020}
Rudi, N., Olivares, M., and Shetty, A. (2020).
\newblock Ordering sequential competitions to reduce order relevance: Soccer
  penalty shootouts.
\newblock {\em Plos One}, 15(12):e0243786.

\bibitem[UEFA, 2021]{UEFA2021f}
UEFA (2021).
\newblock Abolition of the away goals rule in all {UEFA} club competitions.
\newblock 24 June.
  \url{https://www.uefa.com/returntoplay/news/026a-1298aeb73a7a-5b64cb68d920-1000--abolition-of-the-away-goals-rule-in-all-uefa-club-competitions/}.

\bibitem[Vandebroek et~al., 2018]{VandebroekMcCannVroom2018}
Vandebroek, T.~P., McCann, B.~T., and Vroom, G. (2018).
\newblock Modeling the effects of psychological pressure on first-mover
  advantage in competitive interactions: The case of penalty shoot-outs.
\newblock {\em Journal of Sports Economics}, 19(5):725--754.

\bibitem[Wunderlich et~al., 2020]{WunderlichBergeMemmertRein2020}
Wunderlich, F., Berge, F., Memmert, D., and Rein, R. (2020).
\newblock Almost a lottery: the influence of team strength on success in
  penalty shootouts.
\newblock {\em International Journal of Performance Analysis in Sport},
  20(5):857--869.

\end{thebibliography}

\end{document}